\documentclass[11pt]{article}
\usepackage{setspace}

\usepackage[letterpaper, margin=1in]{geometry}

\usepackage{threeparttable}
 \usepackage{physics}
 \usepackage{varwidth} 
 \usepackage[ruled,linesnumbered]{algorithm2e}

\DeclareMathOperator{\poly}{poly}
\DeclareMathOperator{\row}{Row}

\usepackage[utf8]{inputenc}
\usepackage{amsmath,bm}
\usepackage{amsthm}
\usepackage{amssymb}
\usepackage{mathtools}
\usepackage{bbm}
\allowdisplaybreaks
\usepackage{graphicx}
\usepackage[english]{babel}
\usepackage[font=small]{caption}
\usepackage{subcaption}
\usepackage{wrapfig}
\usepackage{tabularx}
\usepackage{multirow}
\usepackage{multicol}
\usepackage{thmtools, thm-restate}
\usepackage[normalem]{ulem}
\usepackage[noadjust]{cite}
\usepackage{makecell}
\usepackage{hyperref}
\usepackage{hyphenat}
\usepackage{cleveref}

\usepackage{enumitem}

\usepackage[draft]{changes}

\newtheorem{theorem}{Theorem$\!$}

\newtheorem{lemma}[theorem]{Lemma$\!$}
\newtheorem{corollary}[theorem]{Corollary$\!$}

\newtheorem{definition}[theorem]{Definition$\!$}

\newtheorem{condition}[theorem]{Condition$\!$}
\usepackage{graphicx}

\usepackage{xcolor}

\newcommand{\ceil}[1]{\left \lceil #1 \right \rceil }
\newcommand{\floor}[1]{\left \lfloor #1 \right \rfloor }

\newcommand{\cC}{\mathcal{C}}

\newcommand{\bA}{\mathbf{A}}

\newcommand{\bG}{\mathbf{G}}

\newcommand{\bS}{\mathbf{S}}

\newcommand{\bY}{\mathbf{Y}}
\newcommand{\bZ}{\mathbf{Z}}

\newcommand{\mybold}[1]{\bm{#1}}

\newcommand{\bc}{{\mybold{c}}}

\newcommand{\bs}{{\mybold{s}}}

\newcommand{\bu}{{\mybold{u}}}
\newcommand{\bv}{{\mybold{v}}}

\newcommand{\bx}{{\mybold{x}}}
\newcommand{\by}{{\mybold{y}}}
\newcommand{\bz}{{\mybold{z}}}

\begin{document}
%
\title{Linear List Decodable Edit-Correcting Codes with Rate Approaching $1$}
%
%
%

\author{
Yuting~Li\thanks{Yuting Li is with the Department of Computer Science at the University of Virginia, USA, \texttt{mzy8rp@virginia.edu}.},  Ryan~Gabrys, Member,~IEEE\thanks{Ryan Gabrys is with Calit2 at the University of California-San Diego, USA, \texttt{rgabrys@ucsd.edu}.}, Farzad~Farnoud,~Member,~IEEE
\thanks{Farzad Farnoud is with the Department of Electrical and Computer Engineering and the Department of Computer Science at the University of Virginia, USA, \texttt{farzad@virginia.edu}.}
}

\maketitle

\begin{abstract}
Linear codes correcting one deletions have rate at most $1/2$.   
    In this paper, we construct linear list decodable codes correcting edits with rate approaching $1$ and reasonable list size. Our encoder and decoder run in polynomial time.
\end{abstract}


%

\section{Introduction}
Codes correcting edits have been studied intensively.  In \cite{haeupler2017synchronization}, Haeupler et al.\ propose a  technique called `synchronization strings' and  construct codes correcting edits with rate approaching the Singleton bound over large constant alphabet. In \cite{cheng2018deterministic} and \cite{haeupler2019optimal}, Cheng et al.\ and Haeupler respectively construct binary codes correcting $\epsilon$ fraction of edits with rate $1-O(\epsilon\log^2\frac{1}{\epsilon})$.
As for linear codes, in \cite{abdel2007linear}, Abdel-Ghaffar et al.\ show that a linear code correcting one deletion has rate at most $1/2$. In \cite{cheng2022efficient},  Cheng et al.\ establish the \emph{half Singleton bound}, which says that a linear code correcting $\delta$ fraction of edits has rate at most $\frac{1-\delta}{2}$. In \cite{cheng2022efficient}, Cheng et al.\ also construct asymptotically good linear edit-correcting codes. In \cite{con22explicit}, Con et al.\ construct linear codes correcting a fraction of $\delta$ edits with rate $R=\frac{1-4\delta}{8}-\epsilon$ over a field of size $\poly(1/\epsilon)$ and construct binary linear codes correcting a fraction of $\delta$ edits with rate $R=\frac{1-54\delta}{1216}$. In \cite{cheng2023linear}, Cheng et al.\ study linear edit-correcting  codes for high rate and high noise regimes. Specifically, for the high noise regime, for any $\epsilon>0$, they construct linear codes correcting $1-\epsilon$ fraction of edits with rate $\Omega(\epsilon^2)$ over an alphabet of size $\poly(1/\epsilon)$ and construct linear codes correcting $1-\epsilon$ fraction of edits with rate $\Omega(\epsilon^4)$ over an alphabet of size $\poly(1/\epsilon)$ with polynomial time encoder and decoder. For the high rate regime, they construct binary linear codes with rate $1/2-\epsilon$ correcting $\Omega(\epsilon^3\log^{-1} \frac{1}{\epsilon})$  fraction of edits. There are also several results on list decoding against edits. In \cite{Guruswami2017}, Guruswami et al.\ constructed explicit binary codes of rate $\Tilde{\Omega}(\epsilon^3)$ which is list decodable from a $1/2-\epsilon$ fraction of deletions with list size $(1/\epsilon)^{O(\log\log \epsilon)}$ for $0<\epsilon<1/2$. In \cite{Guruswami20optimally},
for any desired $\epsilon>0$ Guruswami et al.\ construct a family of binary codes of positive rate which can be efficiently list-decoded from any combination of $\gamma$ fraction of insertions and $\delta$ fraction of deletions as long as $\gamma+2\delta<1-\epsilon$.

There are also several results on Reed-Solomon codes correcting edits. In \cite{con23reed}, Con et al.\ prove that over fields of size $n^{O(k)}$
there are $[n, k]$ Reed-Solomon codes that
can decode from $n-2k+ 1$ edit errors and hence attain the half-Singleton bound.
They also give a deterministic construction of such codes over much larger fields (of size $n^{k^{
O(k)}}$). In \cite{con24optimal}, Con et al.\  construct $[n,2]$ Reed–Solomon codes
that can correct $n-3$ edit errors with alphabet size $q = O(n^3)$
), which achieves the minimum 
field size needed for such codes.  In \cite{con2024random}, Con et al.\
 prove that with high probability, random Reed-Solomon codes approach
the half-Singleton bound with
linear-sized alphabets. In \cite{beelen2025reed}, Beelen et al.\ show that almost all 2-dimensional Reed-Solomon codes correct at least one 
insertion or deletion. Moreover, for large enough field size $q$, and for any $k\geq 2$, they 
show that there exists a full-length $k$-dimensional Reed-Solomon code that corrects $q/10k$ edits.  They also present a polynomial time algorithm that constructs rate $1/2$ Reed-Solomon codes that can correct a single insertion with alphabet size $q=O(k^4)$.

Recall that linear codes correcting one deletion have rate at most $1/2$. In this paper, we study whether there exist linear list decodable insdel codes with rate approaching $1$. We give an affirmative answer. 
\begin{theorem}[Informal]
    For any $\eta>0$, there exist linear list decodable codes of length $n$ correcting $\eta$ fraction of edits with rate $1-O(\eta^{1/4})$ with 
reasonable list size. The codes have polynomial-time encoders and decoders.
\end{theorem}
\subsection{Notations and Definitions}
We use $d_e(\bx,\by)$ to denote the edit distance between two strings $\bx$ and $\by$, which is the minimum number of insertions and deletions it  needs to change $\bx$ to $\by$.
Let $B(\by,d):=\{\bx \in \{0,1\}^*:d_e(\bx,\by)\leq d\}$. Let $S$ be a set of strings, define $B(S,d)=\bigcup_{\by\in S}B(\by,d)$. Define $B_n(\by,d):=B(\by,d)\cap \{0,1\}^n$.  Sometimes we identify a linear code with its generating matrix. All the vectors in the paper are row vectors. Let $S$ be a set, and $X\xleftarrow{\$}S$ indicates that $X$ is a random variable whose distribution is uniform over $S$. Suppose $\bA$ is a matrix, we use $\row \bA$ to denote the row space of $\bA$. We define the concept of list decoding in the edit distance context.
\begin{definition}
Suppose $\delta\in [0,1]$ and $L$ is a positive integer. 
    A code $\cC\subset \mathbb{F}_{\textcolor{blue}{q}}^n$ is $(\delta,L)$-list decodable (with respect to edits) if for all $\by \in \{0,1\}^*$, $\abs{B_n(\by,\delta n)\cap \cC}\leq L$.
\end{definition}
In our construction, we use  list recoverable codes as outer codes. Here we give the definition of list recoverable codes. 
\begin{definition}
Let $1\leq l\leq L$, $\alpha\in [0,1]$.
    A code $C\subset \mathbb{F}_q^n$ is $(\alpha,l,L)$-list recoverable if, for all $S_1,S_2,\dotsc,S_n\subset \mathbb{F}_q$ with $\abs{S_i}\leq l$, there are at most $L$ codewords $c\in C$ so that $\abs{\{i\in[n]:c_i\notin S_i\}}\leq \alpha n$.
\end{definition}

\subsection{Our Approach}
The general approach leverages code concatenation. Suppose we have encoded our message (via a linear encoder $E'$) to a codeword $\bc'=(c_1',c_2',\dotsc,c_n') \in \mathbb{F}^n_{2^a}$  
for some positive integer $a$. By isomorphism  $\mathbb{F}_{2^a}\simeq \mathbb{F}_2^a$, we can also write $\bc'=(\bc_1',\bc_2',\dotsc,\bc_n') \in \mathbb{F}_2^{an}$. 
To further protect it from insertions and deletions, 
we make use of the codes $\bS_1, \ldots, \bS_n$ to obtain the final codeword
$(\bc_1,\bc_2,\dotsc,\bc_n)=(\bc_1'\bS_1,\bc_2'\bS_2,\dotsc,\bc_n'\bS_n) \in \mathbb{F}_2^{bn}$, where $\bS_i\in \mathbb{F}_2^{a\times b}$. For shorthand, we will refer to each $\bc_i \in \mathbb{F}_2^b$ in our codeword as a block. 
Note that the whole process is linear over $\mathbb{F}_2$.

We now describe the main ideas of our construction. Suppose $\bc'=(\bc_1',\bc_2',\dotsc,\bc_n')$, and 
$(\bc_1,\bc_2,\dotsc,\bc_n)=(\bc_1'\bS_1,\bc_2'\bS_2,\dotsc,\bc_n'\bS_n)$ 
is transmitted. We first assume that only block deletions happen so that for each block $\bc_i$ that is transmitted either it appears in the output or the entire block is deleted.
So, we can assume the decoder receives $(\bc_{i_1},\bc_{i_2},\dotsc,\bc_{i_m})$ where $\{i_1, i_2, \ldots, i_m \} \subseteq [n]$ and $i_1 < i_2 < \cdots < i_m$. As a warm-up, we first consider the problem of trying to align each block of the received sequence to the one transmitted under the following condition.
\begin{condition}\label{condition:2}
    $\row \bS_i\cap \row  \bS_j=0$ for $i\neq j$.
\end{condition}
If Condition~\ref{condition:2} holds, then clearly the receiver can uniquely determine $i_j$ (assuming $\bc_{i_j}$ is non-zero\footnote{For the purposes of brevity, we will assume throughout this section that $\bc_{i_j}\neq 0$.}), since $i_j$ is the unique $i$ satisfying  $\bc_{i_j}\in\row \bS_{i}$. However, it can be shown that Condition~\ref{condition:2} implies that $b\geq 2a$, which means that the rate of $\bS_{i}$'s is at most $\frac{1}{2}$, and so the rate of the resulting concatenated code is at most $\frac{1}{2}$. Since we want to get codes with rate approaching $1$, we will relax the previous condition and next consider the following one instead.

\begin{condition}\label{condition:l}Let $l$ be a positive integer greater than $2$. Then, 
    $\row \bS_{k_1}\cap \row \bS_{k_2}\cap\cdots \cap\row \bS_{k_{l+1}}=0$ for $1\leq k_1<k_2< \dotsc < k_{l+1}\leq n$.
\end{condition}
Condition~\ref{condition:l} is a natural generalization of Condition~\ref{condition:2}. Since the space $\row \bS_{k_1}\cap \row \bS_{k_2}\cap\cdots \cap\row \bS_{k_{l+1}}$ has dimension at least $(\ell+1)a - \ell b$, it follows that the rate of the $\bS_i$'s is at most $\frac{1}{1+1/l}$. The key property here, and one which we will need for the construction, is that as $l$ gets larger, $\frac{1}{1+1/l}$ approaches 1. If Condition~\ref{condition:l} holds, then the decoder can pin down any received block $i_j$ to at most $l$ possibilities. 
Specifically, these possible $i$'s are the ones satisfying $\bc_{i_j}\in\row\bS_{i}$. 

To clarify the decoding process, we will use a collection of decoding boxes,  $B_1, \ldots, B_n$, where each box corresponds to a block in our codewords. Intuitively, the decoding box $B_i \subseteq \mathbb{F}_2^a$ contains a set of vectors from $\mathbb{F}_2^{a}$ that represent possible guesses for the value of each $\bc'_i$ that was used during the encoding process.
In order to generate the set $B_i$, we follow the following procedure. If $i$ is one of the possibilities for the received block $i_j$, then the decoder puts $\bc_{i_j}\bS_i^{-1}$ into $B_i$ (here we abuse notation since $\bS_i$ has full row rank and so $\bc_{i_j}\bS_i^{-1}$ can be uniquely determined). 
We have two observations. First, $\bc_{i_j}\in B_{i_j}$. Second, the total number of items in all the boxes are bounded. Therefore, it may be possible to show the even stronger condition that $\bc_i'\in B_i$ for all but a small fraction  of $i$'s.
Indeed, if this property holds, then if we choose $E'$ to be a list recoverable encoder, then we can output a list that contains $c'$.

Next, we refine our focus on the condition where deletions in the codewords occur such that each block is either entirely deleted or remains fully intact, under which the decoding procedure operates as follows. In this scenario, the decoder may receive a corrupted version of $\bc_{i}$ and the alignment between distinct blocks may not be the same throughout the received vector. To handle the more general framework, assume that the decoder receives the string $\bs$. The decoder will scan $\bs$ from left to right and will examine substrings, which are contiguous windows of symbols, of $\bs$ of length $b$ that are $t$ coordinates (or steps) away from each other for some fixed positive integer $t$. If there exists a $\bc_i$ that has experienced a small fraction of edits, then during the scanning process the receiver will encounter a substring $\by$ from $\bs$ where $d_e(\by,\bc_i)$ is small. Now, suppose we have a substring $\by$ satisfying $d_e(\by,\bc_{i})\leq \delta b$ for some $\delta\in (0,1)$. We want to come up with a condition  that allows for $\bc'_{i}$ to appear in the decoding box $B_{i}$, but $\bc'_i$ also does not appear in too many other decoding boxes. In order to guarantee these properties, we will work with the following generalization of Condition~\ref{condition:l}.


\begin{condition}[$(\delta,l,L)$-sync]\label{condition:delta}
    Let $l$ be a positive integer greater than $2$.
    $$B((\row \bS_{k_1}\setminus\{0\}),\delta b)\cap  B((\row \bS_{k_2}\setminus\{0\}),\delta b)\cap\cdots \cap B((\row \bS_{k_{l+1}}\setminus\{0\}),\delta b)=\emptyset$$ for $1\leq k_1<k_2< \dotsc < k_{l+1}\leq n$. And $\bS_i$ is $(\delta,L)$-list decodable for $i\in [n]$.
\end{condition}
 
Under this more general condition, we now discuss and examine the decoding process in a little more detail. First, the decoder will initialize each decoding box $B_i$ to contain the all-zeros vector (since it has an intersection with each code due to linearity). Define $A_i(\by) :=\{\bx\neq 0:\by\in B(\bx\bS_i,\delta b)\}$ for $i\in[n]$. Following the procedure described in the previous paragraph, for each substring $\by$ in $\bs$, the decoder will add the strings in $A_i(\by)$ to $B_i$. 
If Condition~\ref{condition:delta} (which we call $(\delta,l,L)$-sync) holds, then $A_i(\by)$ is non-empty for at most $\ell$ choices of $i$, and each $A_i(\by)$ has size at most $L$. These observations will help us to eventually bound the size of $B_i$'s.
We will later show in \Cref{lem:sequenceexist} that there exists $\{\bS_i\}_{i}$ that satisfies Condition~\ref{condition:delta} with rate approaching $1$ for specific parameters.

\section{Preliminaries}
We will make use  of small-bias distributions, which is defined below.

\begin{definition}
    Let $n$ be a positive integer and $\epsilon>0$. Let $\bY \sim Y$ be a random vector over $\mathbb{F}_2^n$ distributed according to $Y$. We call $Y$    $\epsilon$-biased   if for all non-zero $\bx \in\mathbb{F}_2^n$, the statistical distance between $\bY \bx^T$ and the uniform distribution over $\mathbb{F}_2$ is at most $\epsilon$. Let $s$ be a positive integer and $\bZ$ be a uniform random vector over $\mathbb{F}_2^s$. We call $g:\mathbb{F}_2^s\rightarrow \mathbb{F}_2^n$ an $\epsilon$-biased generator if $g(\bZ)$ is  $\epsilon$-biased.
\end{definition}

\begin{lemma}[\cite{alon1992simple}]\label{lem:smallbiased}
    For $\epsilon>0$, there exists an explicit $\epsilon$-biased generator $g:\mathbb{F}_2^s\rightarrow\mathbb{F}_2^n$, where $s=O(\log n +\log \frac{1}{\epsilon})$.
\end{lemma}

The following Lemma bounds the statistical distance between a small biased distribution and the uniform distribution.                 
\begin{lemma}[Vazirani's XOR Lemma \cite{goldreich2011three}] Let $X$ be an $\epsilon$-biased distribution over $\mathbb{F}_2^n$, and let $U$ denote the uniform distribution over $\mathbb{F}_2^n$. The statistical distance between $X$ and $U$ is at most $\epsilon \cdot 2^{n/2}$.
\end{lemma}

The next lemma bounds the size of a ball with respect to edit distance.
\begin{lemma}
For $\delta\leq 1/2$, a positive integer $n$, and $\by\in\{0,1\}^*$, we have     $\abs{B_n(\by,\delta n)}\leq 2^{5H(\delta)n}$.
\end{lemma}

\begin{proof} 
We first assume that $\abs{\by}\in [n-\delta n,n+\delta n]$.
Every string $\bx$ in $B_n(\by,\delta n)$ can be obtained in the following way. We first choose $\lfloor\delta n\rfloor$ positions from $\abs{\by}$ positions of $\by$. For each position, we can either delete the symbol in this position or not. This will result in some string $\bz$. We then insert $i =n-\abs{\bz}\le \delta n$ bits in $\bz$ to obtain $\bx$. There are   $\binom{n}{i}2^i\le \binom{n}{\delta n}2^{\delta n}$ ways to get $\bx$ from $\bz$. This is because there are $\binom{n}{i}$ ways to choose which elements of $\bx$ are inserted and $2^i$ ways to choose what they are. Therefore, we have

\begin{align*}
     \abs{B_n(\by,\delta n)}&\leq \binom{\abs{\by}}{\delta n}2^{\delta n} \binom{n}{\delta n}2^{\delta n}\\
     &\leq \binom{n+\delta n}{\delta n}^22^{2\delta n}\\
     &\leq 2^{2(1+\delta)nH\left(\frac{\delta}{1+\delta}\right)}2^{2\delta n}.
\end{align*}

If $\abs{\by}<n-\delta n$ or $\abs{\by}>n+\delta n$, then $\abs{B_n(\by,\delta n)}=0$. 
Note that
\begin{align*}(1+\delta)H\left(\frac{\delta}{1+\delta}\right)&\leq (1+\delta)\frac{\delta}{1+\delta}\log \frac{1+\delta}{\delta}+(1+\delta)\frac{1}{1+\delta}\log (1+\delta)\\
&\leq \delta\log \frac{1+\delta}{\delta}+\log (1+\delta)\\
&\leq \delta\log \frac{2}{\delta}+\log \frac{1}{1-\delta}\\
&\leq 2\delta\log \frac{1}{\delta}+2(1-\delta)\log \frac{1}{1-\delta}\\
&\leq 2H(\delta).
\end{align*}
Thus, 

\begin{align*}
     \abs{B_n(\by,\delta n)}&\leq 2^{4H(\delta)n+2\delta n}\leq 2^{5H(\delta)n},
\end{align*}
where the last step follows from the fact that $2x\le H(x)$ for $0\le x\le 1/2$.
\end{proof}

\section{List Decoding Capacity of Random Linear Codes}

We study the list decoding capacity of random linear codes.
\begin{lemma}\label{lem:randomlist}
    Let $\delta\in (0,1/2)$,  $\epsilon>0$,  $L=\ceil{2^{2/\epsilon+1}}$, and $R=1-5H(\delta)-\epsilon$. If $\frac{k}{n}\leq R$ and $\bG \xleftarrow[]{\$}\mathbb{F}_2^{k\times n}$, then
    \begin{align*}
      \Pr[\bG \text{ is not }(\delta,L)\text{-list decodable}]\leq 2^{-0.5n+1}.  
    \end{align*}
\end{lemma}

\begin{proof}
    If $\bG$ is not $(\delta,L)$-list decodable,  then there exist $\by\in \mathbb{F}_2^{*}$ and  distinct non-zero $\bx_1,\dotsc,\bx_{L}\in \mathbb{F}_2^{k}$ such that $\bx_{i}\bG \in  B_n(\by,\delta n)$ for all $i\in [L]$. 
    Note that there exists a subset $\{\bu_1,\dotsc,\bu_{\floor{\log L}}\}$ of size $\floor{\log L}$  of  $\{\bx_1,\dotsc,\bx_{L}\}$ that are linearly independent. Thus,

    \begin{align*}
        &\Pr[\bG \text{ is not } (\delta,L) \text{ list decodable}]\\
        &\leq \Pr[\exists \by\in \{0,1\}^* \text{ and linearly independent }\bu_1,\dotsc,\bu_{\floor{\log L}}:\bu_{i}\bG \in B_n(\by,\delta n)]\\
        & \leq\underbrace{2^{(1+\delta)n+1}}_{\text{choice of }\by} \underbrace{2^{ Rn\floor{\log L}}}_{\text{choice of }\bu_i's}\left(\frac{\underbrace{2^{5H(\delta)n}}_{\text{size of }B_n(\by,\delta n)}}{2^n}\right)^{\floor{\log L}}\\
        &\leq 2^{((1+\delta)+(R+5H(\delta)-1)\frac{2}{\epsilon})n+1}\\
        &\leq 2^{(\delta-1)n+1}\\
        &\leq 2^{-0.5n+1}.
    \end{align*}
\end{proof}
Note that in the Hamming distance scenario, a random linear code in $\mathbb{F}_2^n$ of rate $1-H(\delta)-\epsilon$ can be list decoded from a fraction of $\delta$ errors with high probability with lists of size at most $O(1/\epsilon)$ \cite{guruswami2010list}. However, in the edit distance scenario, we do not know whether the $2^{O(1/\epsilon)}$ list size in \Cref{lem:randomlist} can be improved.

\section{Sync Matrix Sequences}
In this section, we define what we call sync sequences. In the following definition, we rewrite Condition~\ref{condition:delta} explicitly.
\begin{definition}\label{def:sync}
    We call a sequence of binary matrices $\bS_1,\bS_2,\dotsc,\bS_n$  (of size $a\times b$) $(\delta,l)$-sync if there does not exist $1\leq i_1<i_2<\dotsc<i_{l+1}\leq n$,  non-zero $\bx_1,\dotsc,\bx_{l+1}\in \mathbb{F}_2^a$ and $\bv\in \{0,1\}^{*}$, such that $\bx_j\bS_{i_j}\in B_b(\bv,\delta b)$ holds for $j\in [l+1]$. We call $\bS_1,\bS_2,\dotsc,\bS_n$  $(\delta,l,L)$-sync if \begin{enumerate}
    \item $\bS_1,\bS_2,\dotsc,\bS_n$ is  $(\delta,l)$-sync.
    \item  $\bS_i$ is $(\delta,L)$-list decodable for $i\in [n]$.
    \item $\bS_i$ has full rank for $i\in[n]$.       
\end{enumerate}
\end{definition}
The next lemma shows that a random matrix sequence of specific parameters has high probability to be $(\delta,l,
L)$-sync.

\begin{lemma}\label{lem:sequenceexist}
Let $\delta\in (0,1/2)$,  $l$ be a positive integer, $b\geq 4(l+1)\log n$, $L=2^{l+2}$, and  $R=1-\frac{2}{l+1}-5H(\delta)$.   If $a\leq Rb$, and  $\bS_1\xleftarrow[]{\$}\mathbb{F}_2^{a\times b},\dotsc,\bS_n\xleftarrow[]{\$}\mathbb{F}_2^{a\times b}$ are independent, then  $\bS_1,\bS_2,\dotsc,\bS_n$  is $(\delta,l,L)$-sync with  probability at least $1-3/n$ for sufficiently large $n$.
\end{lemma}

\begin{proof}
Recall that $\bS_1,\bS_2,\dotsc,\bS_n$ is $(\delta,l,L)$-sync if it satisfies the three conditions in \Cref{def:sync}. It suffices to bound the failure probability of 1), 2), and 3)  separately.

Note that 1) fails implies that there exist $1\leq i_1<i_2<\dotsc<i_{l+1}\leq n$,  non-zero $\bx_1,\dotsc,\bx_{l+1}\in \mathbb{F}_2^a$ and $\bv\in \{0,1\}^{*}$, where $\abs{\bv}\in[(1-\delta)b,(1+\delta)b]$, such that $\bx_j\bS_{i_j}\in B_b(\bv,\delta b)$ for $j\in[l+1]$. So we have 
\begin{align*}
    \Pr[1)\text{ fails}]&\leq \underbrace{n^{l+1}}_{\text{choice of }i_j's}\underbrace{2^{(l+1)a}}_{\text{choice of }\bx_j's}\underbrace{2^{(1+\delta)b+1}}_{\text{choice of }\bv}\left(\frac{\underbrace{2^{5H(\delta)b}}_{\text{size of }B_b(\bv,\delta b)}}{2^b}\right)^{l+1} \nonumber\\
    &\leq n^{l+1}2^{(l+1)b\left(R+\frac{1+\delta}{l+1}+5H(\delta)-1\right)+1}\\
    &\leq n^{l+1}2^{(l+1)b\left(\frac{\delta-1}{l+1}\right)+1}\\
    &\leq n^{l+1}2^{-b/2+1}\\
    &\leq 1/n.
\end{align*}
Similarly, 2) fails implies that there exists $i\in[n]$ such that $\bS_i$ is not $(\delta,L)$-list decodable.
 Take $b$ and $\frac{2}{l+1}$ to be the $n$ and $\epsilon$ in \Cref{lem:randomlist}, 
then we have 
\begin{align*}
    \Pr[2)\text{ fails}]\leq n 2^{-0.5b+1}\leq 1/n.
\end{align*}
Condition 3) fails implies that there exists $i\in[n]$ and non-zero $\bx\in \mathbb{F}_2^a$ such that $\bx\bS_i=0$. 
\begin{align*}
    \Pr[3) \text{ fails}]&\leq \frac{n2^a}{2^b}\\
    &\leq \frac{n}{2^{(1-R)b}}\\
    &\leq \frac{n}{2^{\frac{2b}{l+1}}}\\
    &\leq 1/n.
\end{align*}
Therefore, $\bS_1,\bS_2,\dotsc,\bS_n$ is $(\delta,l,L)$-sync with  probability at least $1-3/n$.
\end{proof}

\subsection{Construction of Sync  Matrix Sequences}
We first show that it is easy to verify that a sequence is $(\delta,l,L)$-sync.
\begin{lemma}\label{lem:versync}
    For $0<\delta<1$, positive integer $l$ and $L$, given a sequence of binary matrices $\bS_1,\bS_2,\dotsc,\bS_n$ of size $a\times b$, where $a,b=O(\log n)$, it takes polynomial time to verify whether $\bS_1,\bS_2,\dotsc,\bS_n$ is  $(\delta,l,L)$-sync.
    
\end{lemma}

\begin{proof}
    It suffices to verify the three conditions in \Cref{def:sync}. For the first condition, one needs to verify that for each $1\leq i_1<i_2<\dotsc<i_{l+1}\leq n$, $\bx_{1},\bx_2,\dotsc,\bx_{l+1}\in \mathbb{F}_2^a$, and $\bv\in\{0,1\}^*$, where $\abs{\bv}\in[(1-\delta)b,(1+\delta)b]$, whether $\bx_j\bS_{i_j}\in B_b(\bv,\delta b)$ holds for all $j\in [l+1]$. There are $O(n^{l+1})$ choices for  $1\leq i_1<i_2<\dotsc<i_{l+1}\leq n$, $2^{(l+1)a}$ choices for $\bx_{1},\bx_2,\dotsc,\bx_{l+1}\in \mathbb{F}_2^a$, and $2^{O(b)}=\poly(n)$ choices for $\bv$. Hence, it takes $\poly(n)$ time to verify the first condition.

    For the second condition, one needs to verify that for each $1\leq i\leq n$, $\bv\in\{0,1\}^*$, where $\abs{\bv}\in [(1-\delta )b,(1+\delta)b]$ whether $\abs{\{\bx\in \mathbb{F}_2^a:\bx\bS_i\in B_b(\bv,\delta b)\}}\leq L$. There are $O(n)$ choices for $i$, $2^{O(b)}=\poly(n)$ choices for $\bv$, and computing $\{\bx\in \mathbb{F}_2^a:\bx\bS_i\in B_b(\bv,\delta b)\}$ takes $\poly(n)$ time since $a=O(\log n)$. Hence, it takes $\poly(n)$ time to verify the second condition.
    It is clear that the third condition also takes $\poly(n)$ time to verify.
\end{proof}
  We will use the next lemma to construct sync matrix sequence. 
 
\begin{lemma}\label{lem:constructsequence}
Let $\delta\in (0,1/2)$, $l$ be a positive integer, $b\geq 4(l+1)\log n$,  $L=2^{l+1}$,   $R=1-\frac{2}{l+1}-5H(\delta)$, 
and $a\leq Rb$. Let $g:\{0,1\}^{\phi}\rightarrow \{0,1\}^{a b}$ be a $2^{-2b(l+1)}$-biased generator.  Let  $r_1\xleftarrow{\$} \{0,1\}^{\phi},\dotsc, r_n\xleftarrow{\$}\{0,1\}^{\phi}$ be $(l+1)$-wise independent (any $l+1$ of them are independent).
Let $\bS_i=g(r_i)$ for $i\in [n]$ (treat $\{0,1\}^{ab}$ as the set of $a\times b$ matrices), then $\bS_1,\bS_2,\dotsc,\bS_n$ is  $(\delta,l,L)$-sync with  probability at least $1-n^{-(l+1)}2^{l+3}-2n^{-7}$ for sufficiently large $n$.
\end{lemma}

\begin{proof}
The proof is similar to the proof of \Cref{lem:sequenceexist}. It suffice to bound the failure probability for the three  conditions in \Cref{def:sync}, namely $\Pr[1)\text{ fails}]$, $\Pr[2)\text{ fails}]$, and $\Pr[3)\text{ fails}]$. 

Condition 1) fails implies that there exist $1\leq i_1<i_2<\dotsc<i_{l+1}\leq n$,  non-zero $\bx_1,\dotsc,\bx_{l+1}\in \mathbb{F}_2^a$ and $\bv\in \{0,1\}^{*}$, where $\abs{\bv}\in[(1-\delta)b,(1+\delta)b]$, such that $\bx_j\bS_{i_j}\in B_b(\bv,\delta b)$ for $j\in[l+1]$. Note that each $\bx_j\bS_{i_j}$ is $2^{-2b(l+1)}$-biased, thus by Vazirani's XOR Lemma, $\bx_j\bS_{i_j}$ is $2^{-2b(l+1)}2^{b/2}$-close to uniform distribution over $\{0,1\}^b$. So for fixed $\bx_j$ and $\bv$, it holds that
\begin{align*}
    \Pr[r_j\xleftarrow{\$}\{0,1\}^{\phi(b)}:\bx_j\bS_{i_j}\in B_b(\bv,\delta b)]&\leq \frac{\abs{B_b(\bv,\delta b)}}{2^b}+2^{-2b(l+1)}2^{b/2}\\
    &\leq \frac{2^{5H(\delta)b}}{2^b}+2^{-2b(l+1)}2^{b/2}\\
    &\leq \frac{2^{5H(\delta)b}}{2^b}+2^{-b}\\
    &\leq \frac{2^{5H(\delta)b+1}}{2^b}.
\end{align*}
Moreover, since $r_1,r_2,\dotsc,r_n$ are $(l+1)$-wise independent, we have 
\begin{align}
    \Pr[1)\text{ fails}]&\leq \underbrace{n^{l+1}}_{\text{choice of }i_j's}\underbrace{2^{(l+1)a}}_{\text{choice of }\bx_j's}\underbrace{2^{(1+\delta)b+1}}_{\text{choice of }\bv}\left(\frac{2^{5H(\delta)b+1}}{2^b}\right)^{l+1}\label{eq:1}\\
    &\leq n^{l+1}2^{(l+1)(Rb+5H(\delta)b+1-b)+(1+\delta)b+1}\nonumber\\
    &\leq n^{l+1}2^{(l+1)(-\frac{2b}{l+1}+1)+(1+\delta)b+1}\nonumber\\
    &\leq n^{l+1}2^{(\delta-1)b+l+2}\nonumber\\
    &\leq n^{l+1}2^{-0.5b+l+2}\nonumber\\
    &\leq n^{-(l+1)}2^{l+2}.\nonumber
\end{align}

Condition 2) fails implies that there exists $i\in[n]$, $\bv\in \{0,1\}^*$, and linearly independent $\bu_1,\bu_2,\dotsc,\bu_{l+1}\in \mathbb{F}_2^a$ such that $\bu_j\bS_i\in B_b(\bv,\delta b)$ for $j\in [l+1]$. Since $\bu_1,\bu_2,\dotsc,\bu_{l+1}$ are linearly independent, the joint distribution $(\bu_1\bS_i,\bu_2\bS_i,\dotsc,\bu_{l+1}\bS_i)$ is $2^{-2b(l+1)}$-biased over $\{0,1\}^{b(l+1)}$. Thus, by Vazirani's XOR lemma, the distribution  $(\bu_1\bS_i,\bu_2\bS_i,\dotsc,\bu_{l+1}\bS_i)$ is $2^{-2b(l+1)}2^{\frac{b(l+1)}{2}}$-close to (thus $2^{-b(l+1)}$ close to) the uniform distribution over $\{0,1\}^{b(l+1)}$. Hence,
\begin{align*}
    \Pr[\bu_j\bS_i\in B_b(\bv,\delta b) \text{ for }j\in [l+1]]&\leq \left(\frac{\abs{B_b(\bv,\delta b)}}{2^b}\right)^{l+1}+2^{-b(l+1)}\\
    &\leq \left(\frac{2^{5H(\delta)b}+1}{2^b}\right)^{l+1}\\
    &\leq \left(\frac{2^{5H(\delta)b+1}}{2^b}\right)^{l+1}.
\end{align*}
Thus, 
\begin{align}
    \Pr[2)\text{ fails}]&\leq \underbrace{n}_{\text{choice of }i}\underbrace{2^{a(l+1)}}_{\text{choice of }\bu_j's}\underbrace{2^{(1+\delta)b+1}}_{\text{choice of }\bv}\left(\frac{2^{5H(\delta)b+1}}{2^b}\right)^{l+1}\nonumber\\
    &\leq  n2^{a(l+1)}2^{(1+\delta)b+1}\left(\frac{2^{5H(\delta)b+1}}{2^b}\right)^{l+1}\label{eq:2}\\
    &\leq  n^{-(l+1)}2^{l+2}, \nonumber
\end{align}
where the last inequality comes from comparing \eqref{eq:1} and \eqref{eq:2}.

Condition 3) fails implies that there exists $i\in[n]$ and non-zero $\bx\in\mathbb{F}_2^a$ such that $\bx\bS_i=0$. Note that for a fixed $i$ and a fixed  non-zero $\bx$, $\bx \bS_i$ is a $2^{-2b(l+1)}$-biased distribution over $\{0,1\}^b$. By Vazirani's XOR Lemma, the statistical distance between $\bx \bS_i$ and the uniform distribution is at most $2^{-2b(l+1)}2^{b/2}$. So, 
\begin{align*}
    \Pr[\bx\bS_i=0]&\leq 2^{-b}+2^{-2b(l+1)}2^{b/2}\\
    &\leq 2^{-b}+2^{-b}\\
    &\leq 2^{-b+1}.
\end{align*}
Thus, 
\begin{align*}
    \Pr[3) \text{ fails}]&\leq n 2^{a}2^{-b+1}\\
    &\leq n2^{-(1-R)b+1}\\
    &\leq n2^{-\frac{2b}{l+1}+1}\\
    &\leq 2n^{-7}.
\end{align*}

So, $\bS_1,\bS_2,\dotsc,\bS_n$ is  $(\delta,l,L)$-sync with  probability at least $1-n^{-(l+1)}2^{l+3}-2n^{-7}$.
\end{proof}

\begin{corollary}\label{cor:syncconstuction}
Let $\delta\in (0,1/2)$, $l$ be a positive integer, $4(l+1)\log n\leq b=O(\log n)$,  $L=2^{l+1}$,   $R=1-\frac{2}{l+1}-5H(\delta)$, 
and $a\leq Rb$.
    It takes $\poly(n)$ time  to construct a $(\delta,l,L)$-sync sequence $\bS_1,\bS_2,\dotsc,\bS_n$, where $\bS_i\in \mathbb{F}_2^{a\times b}$ for $i\in [n]$.
\end{corollary}

\begin{proof}
By \Cref{lem:smallbiased}, there exists a $2^{-2b(l+1)}$-biased generator $g:\{0,1\}^{\phi(b)}\rightarrow \{0,1\}^{a\times b}$, where $\phi(b)=O(b)$. Take $\phi=\phi(b)$ in \Cref{lem:constructsequence}. 

    One needs $O(\log n+b)=O(\log n)$ random bits to generate the $(l+1)$-wise independent $r_1\xleftarrow{\$} \{0,1\}^{\phi},\dotsc, r_n\xleftarrow{\$}\{0,1\}^{\phi}$ in \Cref{lem:constructsequence}(\cite{vadhan2012pseudorandomness} Corollary 3.34). So one can try all the $2^{O(\log n)}=\poly(n)$ possibilities in \Cref{lem:constructsequence}. For each possibility, by Lemma~\ref{lem:versync}, it takes $\poly(n)$ time to verify whether it is $(\delta,l,
    L)$-sync. Therefore, it takes $\poly(n)$ time to construct a $(\delta,l,L)$-sync sequence.
\end{proof}

\section{Code Construction}
In this section, we construct list decodable codes correcting insdels with rate approaching $1$. We use codes concatenation.
For parameter $0<\gamma<1/8$, we define the inner codes and outer codes as follows.

\textbf{Inner Codes}
Let $\delta=\delta(\gamma)=4\gamma$, $l=l(\gamma)=\ceil{\frac{1}{\gamma}}-1$,  $R=R(\gamma)=1-\frac{2}{l+1}-5H(\delta)$, $b=b(\gamma,n)=\ceil{ 4(l+1)\log n}$, $a=a(\gamma,n)=\floor{Rb}$, $L=L(\gamma)=2^{l+1}$.
By \Cref{cor:syncconstuction}, we can construct a 
$(4\gamma,l,L)$-sync sequence  $\bS_1,\bS_2,\dotsc,\bS_n$, where $\bS_i\in \mathbb{F}_2^{a\times b}$ for $i\in[n]$. Moreover, the construction takes polynomial time.

\textbf{Outer Codes}
We will leverage the following list recoverable codes.
\begin{theorem}[\cite{hemenway2019local} Theorem A.1]\label{thm:a7}
There are constant $c,c_0,c_1$ so that the following holds. Choose $\epsilon>0$ and a positive integer $l_0$. Suppose that $q\geq l_0^{c/c_1\sqrt{\epsilon}}$ is an even power of $2$. Let $N_0=q^{c_0l_0/c_1\sqrt{\epsilon}}$. Then for all $N\geq N_0$, there is a  deterministic polynomial-time construction of an $\mathbb{F}_q$-linear code $C:\mathbb{F}_q^{(1-c_1\sqrt{\epsilon})N}\rightarrow \mathbb{F}_q^{N}$ of rate $1-c_1\sqrt{\epsilon}$ which is $(\epsilon,l_0,L_0)$-list recoverable in time $\poly(N,L_0)$, where 
$L_0=\exp(\exp(\exp(O_{\epsilon, l_0, q}(\log ^* N))))$ .
\end{theorem}
 Take  $\epsilon=\epsilon(\gamma)=2\gamma$, $l_0=l_0(\gamma)=\ceil{\frac{L(\gamma)}{\gamma^3}}$. Take $q=q(\gamma)$ to be the smallest even power of 2 that is at least  $l_0^{c/c_1\sqrt{\epsilon}}$. Take $N=N(\gamma,n)=a(\gamma,n)n/\log q(\gamma)$, where $a(\gamma,n)$ is a multiple of $\log q(\gamma)$. (Note that we can choose $n$ such that $a(\gamma,n)$ is a multiple of $\log q(\gamma)$ since $a(\gamma,i+1)-a(\gamma,i)\leq 1$ for large enough $i$'s.)
By   Theorem~\ref{thm:a7}, we get a $\mathbb{F}_2$-linear $(2\gamma,l_0(\gamma),L_0(\gamma,n))$-list recoverable code of length $N$ with alphabet size $q(\gamma)$ where $L_0=L_0(\gamma,n)=\exp(\exp(\exp(O_{\gamma}(\log ^* N))))$.
Fold this code to an $\mathbb{F}_2$-linear $(2\gamma,l_0(\gamma),L_0(\gamma,n))$-list recoverable code of length $n$ with alphabet size $2^{a}$. This is our outer code.

\textbf{Encoding}
For a message $m$, we first encode it with the outer code and get $(c_1',c_2',\dotsc,c_{n}')$, where $c_i'\in \mathbb{F}_{2^a}$. Then, we view each $c_i'$ as a vector $\bc_i'\in \mathbb{F}_2^a$ and encode it with $\bS_i$. Thus, the final codeword is $(\bc_1,\bc_2,\dotsc,\bc_n)=(\bc_1'\bS_1,\bc_2'\bS_2,\dotsc,\bc_n'\bS_n)\in \mathbb{F}_2^{bn}$.

\textbf{Decoding}
We now explain the high level idea of the decoder.
 Suppose $\by$ is obtained from $\bx=(\bc_1,\bc_2,\dotsc,\bc_n)$ through at most $\gamma^2bn $ edits. Then there are at most $\gamma n$ $\bc_i$'s that suffer from more than $\gamma b$ edits.  So most of $\bc_i$'s suffer from at most $\gamma b$ edits. Define windows $w_i:=[1+\gamma b i,b+\gamma b i]$ for $i=0,1,\dotsc, n/\gamma$. If $\bc_i$ suffers from at most $\gamma b$  edits, then we can show that $d_e(\bc_i,\by_{w_j})\leq 4\gamma b$ for some $\by_{w_j}$. Since $\bS_1, \bS_2,\dotsc,\bS_n$ is $(4\gamma,l,L)$-sync, we can put very few items into very few boxes to ensure that $\bc_i'$ is in the $i$-th box. Since the outer code is list recoverable, we can output a list that contains $\bx$.
 Note that a similar decoding window method was used in \cite{Guruswami17high}.
\begin{algorithm}
\caption{Decoding}\label{alg:decoding}
\KwData{$\by$, which is a corrupted codeword of $\bx$, satisfying $d_e(\by,\bx)\leq \gamma^2 b n$}
\KwResult{a list $L'$ of size $\leq L_0$ containing $\bx$}
\For{$i=0$ to $n/\gamma$}{
    \For{$j=1$ to $n$}{
        \For{all non-zero $\bz\in\{0,1\}^a$}{
            \If{$d_e(\bz\bS_j,\by_{w_i})\leq 4\gamma b$}{
                put $\bz$ into $\text{box}[j]$                   }
        }
        
    }
}
add $0$ to all the boxes.

\If{any $\text{box}[j]$ has more than $l_0$ elements}{$\text{box}[j]\gets \emptyset$}

run the list recovery algorithm of the outer codes using all the boxes, and get a list $L'$.

\Return{$L'$} 
\end{algorithm}

The decoder has two stages. The first stage is step 1 to 10, which is the  inner codes list decoding  stage. The second stage is step 11 to 15, which is the outer code list recovery  stage.

\begin{theorem}
  For $0<\gamma<\frac{1}{8}$,  the decoder corrects a fraction of $\gamma^2$ edits. That is, if $d_e(\by,\bx)\leq \gamma^2bn$, then when input $\by$, the decoder outputs a list that contains $\bx$.
\end{theorem}

\begin{proof}
    We divide $\bx$ into blocks of length $b$, in total $n$ blocks. We first prove that if block $j$, which is $\bc_j'\bS_j$, suffers from at most $\gamma b$ edits, then $\bc_j'$ is put into box$[j]$ in the first stage.  
    If $\bc_j'=0$, then $\bc_j'$ is put into box$[j]$ by step 10. If $\bc_j'\neq 0$, and block $j$
     becomes a substring in $\by$, say  $\by_{[\alpha,\beta)}$. Then $b-\gamma b\leq \beta-\alpha\leq b+\gamma b$.
    So
    there exists $\alpha'\in[bn]$ such that $d_e(\by_{[\alpha,\beta)},\by_{[\alpha',\alpha'+b)})\leq \gamma b$. Since the step length of windows is $\gamma b$, there exists a window $w_i$ such that $d_e(\by_{[\alpha',\alpha'+b)},\by_{w_i})\leq 2\gamma b$. So  $d_e(\by_{[\alpha,\beta)},\by_{w_i})\leq d_e(\by_{[\alpha,\beta)},\by_{[\alpha',\alpha'+b)})+ d_e(\by_{[\alpha',\alpha'+b)},\by_{w_i})\leq \gamma b+2\gamma b \leq 3\gamma b.$ Thus, $d_e(\bc_j'\bS_j,\by_{w_i})\leq 4\gamma b$ and $\bc_j'$ is put into box$[j]$ in the first stage. 

    Then we analyze how many boxes become empty in the second stage. Note that because the inner codes are $(4\gamma,l,L)$-sync, for each window $w_i$, there are at most $l$ $j$'s that something is put into  box$[j]$ in step 1 to 9. And for each such $j$, there are at most $L$ vectors put into box$[j]$. Therefore the total number of vectors in all the boxes at the end of step 9 is at most $lLn/\gamma$. Because we add $0$ to all the boxes in step $10$, there are at most $lL n/\gamma  l_0$ boxes that become empty in step $11$.

    If a block $j$ suffers from at most $\gamma b$ edits and box $j$ does not become empty in the second stage, then this block contributes to one agreement in the outer code list recovery stage. Since the total number of edits is $\gamma^2 bn$, there are at least $(1-\gamma)n$ blocks suffer from at most $\gamma b$ edits. Thus, the total number of  agreement is at least 
    $
    (1-\gamma)n-lL n/\gamma  l_0$. 
    Note that $lL n/\gamma  l_0\leq l\gamma^2 n\leq \gamma n $, 
    where the first inequality comes from the fact that $l_0=\ceil{\frac{L}{\gamma^3}}$, and the second inequality is because $l=\ceil{\frac{1}{\gamma}}-1\leq \frac{1}{\gamma}$. So, the total number of  agreement is at least $1-2\gamma n$. Since our outer code is $(2\gamma,l_0,L_0(\gamma,n))$-list recoverable, the decoding algorithm is correct.
\end{proof}

Note that the rate of our outer code is $1-c_1\sqrt{2\gamma}$, and the rate of our inner code is $\frac{a}{b}\geq R-1/b=R-o(1)$, where $R=1-\frac{2}{l+1}-5H(\delta)\geq 1-2\gamma-5H(4\gamma)$. Therefore, the rate of our code is at least $(1-c_1\sqrt{2\gamma})(1-2\gamma-5H(4\gamma))=1-O(\sqrt{\gamma})$. The encoding algorithm and decoding algorithm runs in polynomial time. The list size is $L_0=\exp(\exp(\exp(O_{\gamma}(\log ^* n))))$.
In summary, we have the following theorem.
\begin{theorem}
For $0<\eta<\frac{1}{64}$,    there exists a deterministic polynomial time construction of  a linear list decodable code of length $n$ correcting $\eta$ fraction of edits with rate $1-O(\eta^{1/4})$ with list size
$\exp(\exp(\exp(O_{\eta}(\log ^* n))))$. The encoding and decoding time is polynomial in $n$.
\end{theorem}

\section*{Acknowledgment}
This work was supported in part by NSF grants CIF-2312871, CIF-2312873, CIF-2144974, and CCF-2212437.




%
\bibliographystyle{IEEEtran}
\bibliography{bibliofile}



%








\end{document}